\documentclass[a4paper, times, 10pt, twocolumn, twoside,top=3.5cm,bottom=3.5cm,left=2.5cm,right=2cm,]{article}

\usepackage{ISARC}

\usepackage{lscape}
\usepackage{hologo}
\usepackage{breqn}

\usepackage{amsmath}

\usepackage{amsfonts}
\usepackage{enumitem}

\usepackage{appendix}
\newtheorem{theorem}{Theorem}
\newtheorem{assumption}{Assumption}
\newtheorem{remark}{Remark}
\newenvironment{proof}{\paragraph{Proof:}}{$\square$}

\begin{document}

 % Do not change the following line
\linespread{0.5}

\title{Oscillation Reduction for Knuckle Cranes}

\author{Michele Ambrosino$^{a}$, Arnaud Dawans$^b$, Brent Thierens$^{a}$,  Emanuele Garone$^{a}$}

\affiliation{
$^a$Service d’Automatique et d’Analyse des Systèmes, Université Libre de Bruxelles, Brussels, Belgium\\
$^b$Entreprises Jacques Delens S.A., Brussels, Belgium
}

\email{
\href{mailto:e.author1@aa.bb.edu}{Michele.Ambrosino@ulb.ac.be}, 
\href{mailto:e.author1@aa.bb.edu}{adawans@jacquesdelens.be},
\href{mailto:e.author1@aa.bb.edu}{Brent.Thierens@vub.be},
\href{mailto:e.author1@aa.bb.edu}{egarone@ulb.ac.be}
}

% Do not change the following three lines
\maketitle 
\thispagestyle{fancy} 
\pagestyle{fancy}

\begin{abstract}
Boom cranes are among the most common material handling systems due to their simple design. Some boom cranes also have an auxiliary jib connected to the boom with a flexible joint to enhance the maneuverability and increase the workspace of the crane. Such boom cranes are commonly called knuckle boom cranes. Due to their underactuated properties, it is fairly challenging to control knuckle boom cranes. To the best of our knowledge, only a few techniques are present in the literature to control this type of cranes using approximate models of the crane. In this paper we present for the first time a complete mathematical model for this crane where it is possible to control the three rotations of the crane (known as luff, slew, and jib movement), and the cable length. One of the main challenges to control this system is how to reduce the oscillations in an effective way. In this paper we propose a nonlinear control based on energy considerations capable of guiding the crane to desired sets points while effectively reducing load oscillations. The corresponding stability and convergence analysis is proved using the LaSalle’s invariance principle. Simulation results are provided to demonstrate the effectiveness and feasibility of the proposed method.
\end{abstract}


\section{Introduction}\label{sec:Introduction}

Cranes are material handling machines which are used in different industries (i.e., construction, manufacturing, shipbuilding and freight handling) for transporting heavy materials that humans cannot handle. Cranes have the capability of moving the load vertically and horizontally, either along a straight or a curved path. Nowadays, most cranes are manually operated by skilled operators. In this paper we focus on the design of an effective automatic control to obtain accurate positioning and reduce the swing of the load. 

\medskip

Cranes come in various sizes and designs to perform different tasks. Depending on their dynamic properties they can be classified as gantry cranes and rotary cranes. Gantry cranes can be further classified into overhead cranes and container cranes. Rotary cranes can be classified into tower cranes and boom cranes \cite{ref1}.

\medskip

In this paper we will focus on a so called 'knuckle boom' crane which is one of the most common type of boom crane. Boom cranes are characterized by a first boom that can rotate around two orthogonal axes (e.g. slew and luff motions). From the free end of the boom, a payload is suspended using a hoisting rope. The length of the hoisting rope can be driven using a winch. A boom crane can move the payload in the 3D space using the luff and slew movements of the boom and the hoisting of the payload. Such cranes are commonly used in construction sites \cite{boomcntr}. Boom cranes can also have more than one boom. Common variant of the boom crane has an auxiliary jib connected to the boom to enhance the maneuverability. Such boom cranes are the knuckle boom cranes (see Fig.~\ref{fig:crane}). In this paper we present for the first time a complete mathematical model for this kind of crane which takes into account not only the three main rotations (e.g. luff, slew, and jib movement), but also the cable dynamic and the payload oscillations. 

\medskip

% breve accenno al fatto che sono underattuate
As all cranes,  knuckle cranes are nonlinear systems with complicated underactuated dynamics. Underactuated systems \cite{ref2} are commonly found in several areas and applications, such as robotics, aerospace systems, marine systems, flexible systems, mobile systems, and locomotive systems. The condition of underactuation refers to a system with more DoF (number of independent variables that define the system configuration) to be controlled, than actuators (input variables). This restriction implies that some of the configuration variables of the system cannot be directly commanded, which highly complicates the design of control algorithms. In particular for the proposed model of knuckle crane the non-actuated variable are the swing angles of the payload, whereas the four actuated variable will be the three main rotation (i.e. luff, slew, and jib movements) and the length of the cable.

\medskip

%%%% parlare del controllo

Cranes can be controlled using different control laws depending on their operations, which usually involve the process of gripping, lifting, transporting the load, lowering, and releasing the load. A damping capacity of the system plays a significant role towards the precision motion performance. The control schemes developed for cranes can mainly be categorised into open loop and closed loop techniques \cite{ref5}. Control schemes based on a combination of open and closed loop techniques have also been proposed. Input shaping is one of the most used open loop techniques, mostly based on a linearized system, that can be applied in real time, mainly for control of the oscillations of the payload. Anti-swing crane controls using input shaping have been widely implemented in the literature \cite{ref3}-\cite{ref4}. Concerning closed loop techniques, Proportional Integral Derivative (PID) control laws have been proposed for instance in \cite{ref5}, where the authors proposed a position control of a gantry crane system. A state feedback controller has been implemented in \cite{ref7} for a boom
crane system in order to control the load sway angles in the vertical and horizontal boom motions, as well as in the vertical and horizontal boom angles. In \cite{erg} the authors present a constrained control scheme based on the ERG framework for the control of boom cranes. In \cite{ref8} a model predictive control (MPC) approach was used for a boom crane in order to reduce the swing angles as much as possible.  

\medskip
Compared with the other kinds of cranes like boom cranes, the study of knuckle cranes is still at an early stage with much less reported control strategies. In \cite{ref12} the authors focus on controlling mobile electro-hydraulic proportional valves to move the crane to a desired position. In \cite{ref13} the authors solve the problem of controlling knuckle crane through the inverse kinematics without take into account the dynamic of the cable and the payload. In \cite{ref14} an anti-sway control is shown which is performed by simplifying the dynamic model and assuming that the tip of the crane can be controlled directly.
\medskip

%%%% cosa da questo paper:
% modello completo e dettagliato di una knuckle e per la 
% prima volta

To the best of our knowledge, no research has been carried out to develop a detailed mathematical model and develop a control strategy taking into account the strongly nonlinear nature of this type of crane. In this paper we advance the state of the art of the knuckle crane by introducing for the first time a complete mathematical model in which we takes into account all of the degrees of freedom (DoF) that characterize this type of system (i.e. the three rotations, the length of the rope and the payload swing angles) and proposing  novel control strategy designed directly on the nonlinear model. 

\medskip

The rest of this paper is organized as follows. In Section 2,
the dynamic model of the knuckle crane and the control objectives are provided. In Section 3, the proposed
controller is designed, and the corresponding stability analysis is provided in detail. Section 4 shows the results of the simulations regarding the proposed control strategy.

\section{Dynamic Model}

\begin{figure}[ht!]
\centering
\includegraphics[width=0.6\columnwidth]{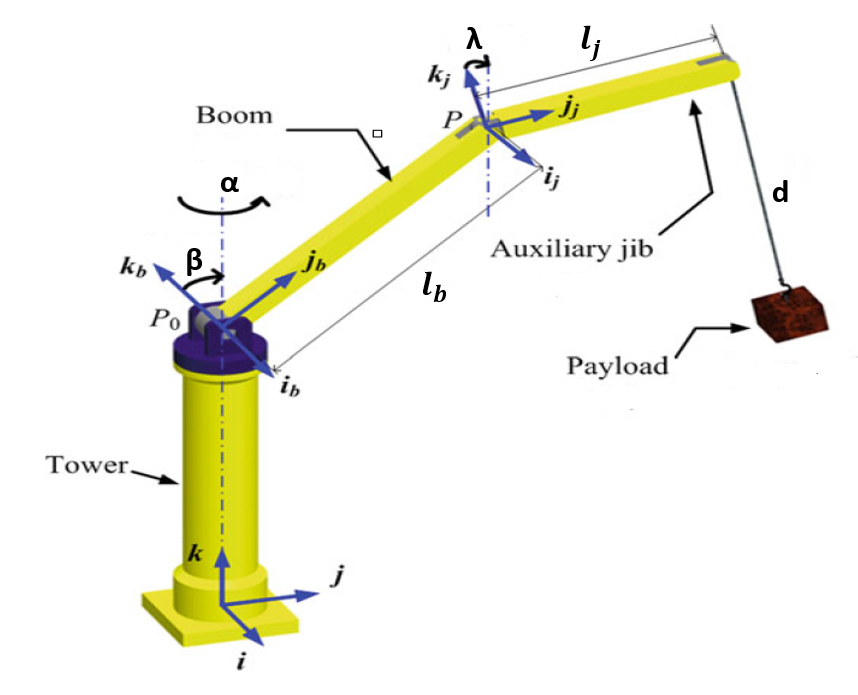}
\caption{\label{fig:crane} Model of a knuckle crane \cite{ref1}.}
\end{figure}

\begin{figure}[ht!]
\centering
\includegraphics[width=0.6\columnwidth]{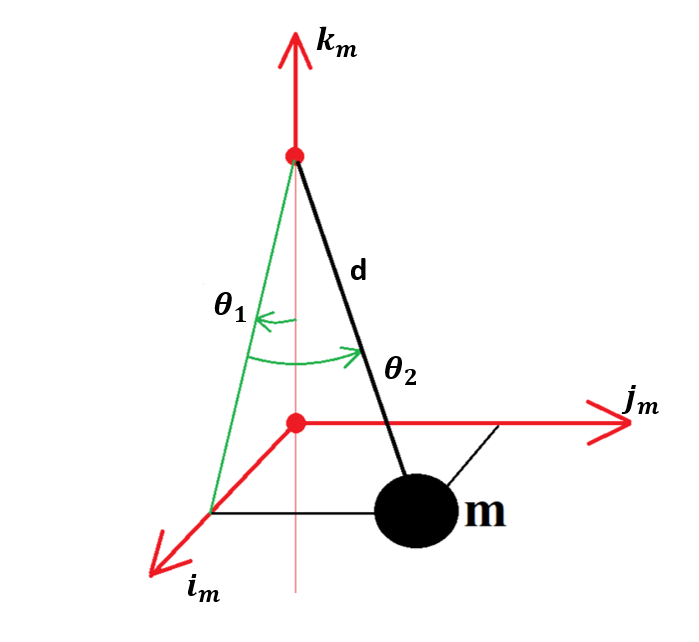}
\caption{\label{fig:pay} Payload swing angles.}
\end{figure}

A schematic view of knuckle crane is shown in Fig.\ref{fig:crane}. The knuckle crane consists of a first boom of length $l_b$ and mass $m_b$ connected to the tower with one revolute joint. The auxiliary jib of length $l_j$ and mass $m_j$ is linked to the first boom by a revolute joint. For the sake of simplicity in this paper, both of these two links are considered to be rigid. The cable is supposed to be massless and rigid, thus the lifting mechanism can be represent as a prismatic joint. The payload of mass $m$ can be represented as lumped mass. The two swing angles of the payload are represented in the Fig.\ref{fig:pay}.

\medskip

The configuration of the crane can be represented by six generalized coordinates, in which, $\alpha$ is the slew angle of
the tower, $\beta$ is the luff angle of the boom, $\gamma$ is the luff angle of the jib, \textit{d} is the length of
the rope, $\theta_1$ is the tangential pendulation mainly due to the motion of the tower and $\theta_2$ is the radial sway mainly due to the motion of the boom. 

\medskip

To simplify the ensuing analysis, the following abbreviations are used: $S_{\alpha} \triangleq \sin(\alpha), S_{\beta} \triangleq sin(\beta), S_{\gamma} \triangleq \sin(\gamma), S_{\theta_1} \triangleq \sin(\theta_1), S_{\theta_2} \triangleq \sin(\theta_2), C_{\alpha} \triangleq \cos(\alpha), C_{\beta} \triangleq \cos(\beta), C_{\gamma} \triangleq \cos(\gamma), C_{\theta_1} \triangleq \cos(\theta_1), C_{\theta_2} \triangleq \cos(\theta_2)$.

\medskip

The dynamic model of the knuckle crane is obtained by using the Lagrange method \cite{ref15}. Firstly, we need to express the system kinematic energy $T(t)$, which consists of three parts, the boom kinematic energy $T_t(t)$, the jib kinematic energy $T_j(t)$, and the payload kinematic energy $T_p(t)$. Then, the system potential energy $U(t)$ consists of the boom gravity energy $U_t(t)$, the jib gravity energy $U_j(t)$, and the payload gravity energy $U_p(t)$. Next, the Lagrange function is constructed as
\begin{equation}\label{eq:lag}
\begin{split}
L(t) = T(t) - U(t) \quad\quad\quad\quad\\ = T_b(t) + T_j(t) + T_p(t) - U_b(t) - U_j(t) - U_p(t),
\end{split}
\end{equation}
where
\begin{eqnarray}\label{eq:T}
&T(t) = {1\over{8}}(m_B((l_BC_{\beta}S_{\alpha}\dot{\alpha} +\nonumber \\ &   l_BC_{\alpha}S_{\beta}\dot{\beta})^2 
+ (l_BC_{\alpha}C_{\beta}\dot{\alpha}- l_BS_{\alpha}S_{\beta}\dot{\beta})^2 
+ l_B^2C_{\beta}^2\dot{\beta}^2)) \nonumber \\  
&
+ {1\over{2}}(m_J((l_BC_{\beta}S_{\alpha}\dot{\alpha} 
+ l_BC_{\alpha}S_{\beta}\dot{\beta} +  {1\over{2}}(l_JC_{\gamma}S_{\alpha}\dot{\alpha}) 
\nonumber\\ &+ {1\over{2}}(l_JC_{\alpha}S_{\gamma}\dot{\gamma}))^2  + (l_BC_{\alpha}C_{\beta}\dot{\alpha} 
+ {1\over{2}}(l_JC_{\alpha}C_{\gamma}\dot{\alpha})  \nonumber \\ &- l_BS_{\alpha}S_{\beta}\dot{\beta} 
- {1\over{2}}(l_JS_{\alpha}S_{\gamma}\dot{\gamma}))^2 - C_{\alpha}C_{\theta_2}\dot{\theta_2}d \nonumber \\  
&+ (l_BC_{\beta}\dot{\beta} 
+ {1\over{2}}(l_JC_{\gamma}\dot{\gamma}))^2)) + l_JC_{\alpha}S_{\gamma}\dot{\gamma}\nonumber \\ & + {1\over{2}}(m((C_{\theta_2}S_{\alpha}S_{\theta_1}\dot{d} 
- C_{\alpha}S_{\theta_2}\dot{d} \nonumber \\  
&+ l_BC_{\beta}S_{\alpha}\dot{\alpha} + l_BC_{\alpha}S_{\beta}\dot{\beta}+ l_JC_{\gamma}S_{\alpha}\dot{\alpha} \nonumber \\  
&  + S_{\alpha}S_{\theta_2}\dot{\alpha}d + C_{\alpha}C_{\theta_2}S_{\theta_1}\dot{\alpha}d + C_{\theta_1}C_{\theta_2}S_{\alpha}\dot{\theta_1}d \nonumber \\ & - S_{\alpha}S_{\theta_1}S_{\theta_2}\dot{\theta_2}d)^2 + (l_BC_{\beta}\dot{\beta} - C_{\theta_1}C_{\theta_2}\dot{d} \nonumber \\ & + l_JC_{\gamma}\dot{\gamma} + C_{\theta_2}S_{\theta_1}\dot{\theta_1}d  + C_{\theta_1}S_{\theta_2}\dot{\theta_2}d)^2 + (S_{\alpha}S_{\theta_2}\dot{d} \nonumber \\ & +  l_BC_{\alpha}C_{\beta}\dot{\alpha} + l_JC_{\alpha}C_{\gamma}\dot{\alpha} - l_BS_{\alpha}S_{\beta}\dot{\beta} - l_JS_{\alpha}S_{\gamma}\dot{\gamma} \nonumber \\ &+ C_{\alpha}S_{\theta_2}\dot{\alpha}d + C_{\theta_2}S_{\alpha}\dot{\theta_2}d + C_{\alpha}C_{\theta_2}S_{\theta_1}\dot{d} \nonumber \\ &+ C_{\alpha}C_{\theta_1}C_{\theta_2}\dot{\theta_1}d - C_{\theta_2}S_{\alpha}S_{\theta_1}\dot{\alpha}d - C_{\alpha}S_{\theta_1}S_{\theta_2}\dot{\theta_2}d)^2)) \nonumber \\ & + {1\over{2}}I_{tot}\dot{\alpha}^2 + {1\over{2}}I_B\dot{\beta}^2  + {1\over{2}}I_J\dot{\gamma}^2, 
\end{eqnarray}

\begin{equation}\label{eq:U}
\begin{split}
U(t) = gm(l_BS_{\beta} + l_JS_{\gamma} - C_{\theta_1}C_{\theta_2}d) 
+ gm_J(l_BS_{\beta} \\+ {1\over{2}}l_JS_{\gamma}) + {1\over{2}}gl_Bm_BS_{\beta},
\end{split}
\end{equation}

The equations of the motion of the crane are derived using the Lagrange's equation
\begin{equation}
\frac{d}{dt}\left(\frac{\partial L(q, \dot{q})}{\partial{\dot{q}}}\right) -  \frac{\partial L(q, \dot{q})}{\partial{q}} = \zeta, 
\end{equation}

where q = $[\alpha, \beta, \gamma, d, \theta_1, \theta_2]^T \in{\mathbb{R}^6}$ represents the system state vector, and $\zeta$ = $[u_1, u_2, u_3, u_4,0,0]^T \in{\mathbb{R}^6}$ represents the control input vector.

The dynamic model of a  knuckle crane (see Fig.~\ref{fig:crane}) can be written as:

\begin{equation}\label{eq:modelmatrix}
{{M(q)\ddot{q} + C(q,\dot{q})\dot{q} +F(\dot{q}) + g(q)} = {\begin{bmatrix}
I_{4x4} \\ 0_{2x2} \end{bmatrix}}u. }
\end{equation}

The matrices M(q) $\in{\mathbb{R}^{6x6}}$,$ C(q,\dot{q})\in{\mathbb{R}^{6x6}}$,$F(\dot{q}) \in{\mathbb{R}^{6}}$ and g(q) $\in{\mathbb{R}^{6}}$ represent the inertia, centripetal-Coriolis, air dynamic friction \cite{ref17}, and gravity.
The system matrices (see Appendix for the detailed description) are defined as follows:

\begin{equation}
M(q) = \begin{bmatrix}
m_{11}&m_{12}&m_{13}&m_{14}&m_{15}&m_{16}\\
m_{21}&m_{22}&m_{23}&m_{24}&m_{25}&m_{26}\\
m_{31}&m_{32}&m_{33}&m_{34}&m_{35}&m_{36}\\
m_{41}&m_{42}&m_{43}&m_{44}&0&0\\
m_{51}&m_{52}&m_{53}&0&m_{55}&0\\
m_{61}&m_{62}&m_{63}&0&0&m_{66}\\
\end{bmatrix}, 
\end{equation}

\begin{equation}
C(q,\dot q) = \begin{bmatrix}
c_{11}&c_{12}&c_{13}&c_{14}&c_{15}&c_{16}\\
c_{21}&c_{22}&c_{23}&c_{24}&c_{25}&c_{26}\\
c_{31}&c_{32}&c_{33}&c_{34}&c_{35}&c_{36}\\
c_{41}&c_{42}&c_{43}&0&c_{45}&c_{46}\\
c_{51}&c_{52}&c_{53}&c_{54}&c_{55}&c_{56}\\
c_{61}&c_{62}&c_{63}&c_{64}&c_{65}&c_{66}\\
\end{bmatrix}, 
\end{equation}

\begin{equation}\label{eq:grav}
{
g(q) = \begin{bmatrix} 0,g_2,g_3,g_4,g_5,g_6
\end{bmatrix}}^T. 
\end{equation}

\begin{equation}\label{eq:F}
{
F(\dot{q}) = \begin{bmatrix} 0,0,0,0,f_1,f_2
\end{bmatrix}}^T, 
\end{equation}

Although the equation of motion (\ref{eq:modelmatrix}) is quite complicated, it has several fundamental properties that can be exploited to facilitate the design of the controller. The two main properties that will be exploited in the next section are:
\begin{enumerate}[label=Property \arabic*.,itemindent=*]
  \item The matrix 
  \begin{equation*}
      {1\over{2}}\dot{M}(q) - C(q,\dot{q}),
  \end{equation*}
  is skew symmetric which means that
  \begin{equation*}
      \zeta^T\left[{1\over{2}}\dot{M}(q) - C(q,\dot{q})\right]\zeta = 0, \quad \zeta \in \mathbb{R}^{6} 
  \end{equation*}
  \item  The gravity vector (\ref{eq:grav}) can be  obtained as the gradient of the crane potential energy (\ref{eq:U}), i.e., $g(q) = \frac{\partial U(q)}{\partial{\dot{q}}}$.
\end{enumerate}

%\begin{remark}
%As one can seen from equation (\ref{eq:modelmatrix}), the knuckle crane is a %typical underactuated system. Infact, there are six DoF to be controlled (e.g. %$\alpha, \beta, \gamma, d, \theta_1, \theta_2)$, with only four inputs %available $u_1, u_2, u_3, u_4$..
%\end{remark}

\subsection{Control objective}

The aim of the control is to move the crane to the desired position and to dampen the swing angles of the load as much as possible. 

\medskip

The control objective can be described mathematically as

\begin{equation}\label{eq:constraints}
\begin{split}
\lim_{t\to \infty}[\alpha(t),\beta(t),\gamma(t),d(t),\theta_1(t),\theta_2(t)] = [\alpha_{d},\beta_{d},\gamma_{d},d_{d},0,0], \\
\lim_{t\to \infty}[\dot{\alpha}(t),\dot{\beta}(t),\dot{\gamma}(t),\dot{d}(t),\dot{\theta_1}(t),\dot{\theta_2}(t)] = [0,0,0,0,0,0],
\end{split}
\end{equation}

where $\alpha_d, \beta_d, \gamma_d, d_d$ are the desired references for the actuated states. 
\newline
In our development we will consider the following reasonable assumptions.
\begin{assumption}\label{ass:1}
The payload swings satisfy the following inequality  $\lvert \theta_{1,2} \rvert\ <{\frac{\pi}{2}}.$
\end{assumption}
\begin{assumption}\label{ass:2}
The cable length is always greater than zero to avoid singularity in the model (\ref{eq:modelmatrix}): $d(t) > ,\forall t\geq0.$
\end{assumption}

\section{Control Design and Stability Analysis}

The control strategy proposed in this paper consists of a nonlinear control law based on energy consideration. The corresponding stability and convergence analysis is demostrated by using the LaSalle’s invariance principle.

\medskip

In order to develop our control law, we started to considered the energy of system (\ref{eq:T})-(\ref{eq:U}), wich is

\begin{equation}\label{eq:E}
    E(t) = {1\over{2}}\dot{q}^TM(q)\dot{q} + mgd(1-C_{\theta_1}C_{\theta_2}),
\end{equation}

where the first term is the kinetic energy of the crane, whereas the second term represents the payload potential energy. Based on (\ref{eq:E}), we can define the following Lyapunov function candidate:
\begin{equation}\label{eq:V}
\begin{split}
V(t) = {1\over{2}}\dot{q}^TM(q)\dot{q} + mgd(1-C_{\theta_1}C_{\theta_2}) \\ +{1\over{2}}k_{p\alpha} e_{\alpha}^2 +{1\over{2}}k_{p\beta} e_{\beta}^2 +{1\over{2}}k_{p\gamma} e_{\gamma}^2 +{1\over{2}}k_{pd} e_{d}^2,
\end{split}
\end{equation}

where $e_{\alpha},e_{\beta},e_{\gamma},e_{d}$ are the error signals defined as:
\begin{equation}
\begin{split}
e_{\alpha} = \alpha_d - \alpha, \quad e_{\beta} = \beta_d - \beta,
e_{\gamma} = \gamma_d - \gamma, \quad e_{d} = d_d - d.
\end{split}
\end{equation}

Differentiating the equation (\ref{eq:V}) with respect
to the time and using (\ref{eq:modelmatrix}) we obtain
\begin{equation}\label{eq:Vdot}
\begin{split}
\dot V(t) = \dot{\alpha}(u_1-k_{p\alpha} e_{\alpha}) \\ +\dot{\beta}(u_2-k_{p\beta} e_{\beta} -gl_BC_{\beta}(m+{1\over{2}}m_Bm_J)) \\ \dot{\gamma}(u_3-k_{p\gamma} e_{\gamma} -gl_JC_{\gamma}(m + {1\over{2}}m_J)) \\ +\dot{d}(u_4-k_{pd} e_{d}+mgC_{\theta_1}C_{\theta_2} +mg(1-C_{\theta_1}C_{\theta_2})) \\-d_{\theta_1}C_{\theta_2}^2 \lvert \theta_1\rvert\ \dot{\theta}_1^2-d_{\theta_2} \lvert \theta_1\rvert\ \dot{\theta}_2^2.
\end{split}
\end{equation}

In order to cancel the gravitational terms and keep $\dot{V}(t)$ non-positive, the following controller is designed:
\begin{equation}\label{eq:u1}
u_1 = k_{p\alpha} e_{\alpha} - k_{d\alpha} \dot{\alpha}, \quad\quad\quad\quad\quad\quad\quad\quad\quad
\end{equation}
\begin{equation}\label{eq:u2}
u_2 = k_{p\beta} e_{\beta} - k_{d\beta} \dot{\beta} + gl_BC_{\beta}(m+{1\over{2}}m_Bm_J), 
\end{equation}
\begin{equation}\label{eq:u3}
u_3 = k_{p\gamma} e_{\gamma} - k_{d\gamma} \dot{\gamma} + gl_JC_{\gamma}(m + {1\over{2}}m_J),\quad\;
\end{equation}
\begin{equation}\label{eq:u4}
u_4 = k_{pd} e_{d} - k_{dd} \dot{d} + mg,\quad\quad\quad\quad\quad\quad\;
\end{equation}

where $k_{p\alpha}$, $k_{p\beta}$, $k_{p\gamma}$, $k_{pd}$, $k_{d\alpha}$, $k_{d\beta}$, $k_{d\gamma}$, $k_{dd}$ $\in\mathbb{R}$ are positive control gains.
\newline
Substituting (\ref{eq:u1})-(\ref{eq:u4}) into (\ref{eq:Vdot}), one obtains
\begin{equation}\label{eq:Vdot2}
\begin{split}
\dot V(t) = -k_{d\alpha} \dot{\alpha}^2 - k_{d\beta} \dot{\beta}^2 -k_{d\gamma} \dot{\gamma}^2 \\ - k_{dd} \dot{d}^2 -d_{\theta_1}C_{\theta_2}^2\dot{\theta}_1^2-d_{\theta_2}\dot{\theta}_2^2 \leq 0,  
\end{split}
\end{equation}

The following theorem describes the stability property of the crane using the proposed control law (\ref{eq:u1})-(\ref{eq:u4}).

\begin{theorem}
Consider the system (\ref{eq:modelmatrix})-(\ref{eq:F}). Under Assumptions 1-2, the controller (\ref{eq:u1})-(\ref{eq:u4}) makes every equilibrium point (\ref{eq:constraints}) satisfying Assumptions 1-2, asymptotically stable. 
\end{theorem}

\begin{proof}

As already seen, the derivative of the Lyapunov function candidate (\ref{eq:V}) is (\ref{eq:Vdot2}) which is negative semidefinite.

At this point let $\Phi$ be defined as the set where $\dot{V}(t)=0$, i.e.
 \begin{equation}
     \Phi = {\{q,\dot{q}|\dot{V}(t) = 0\}}.
 \end{equation}
 
Further, let $\Gamma$ represent the largest invariant set in $\Phi$ where the Assumptions 1-2 are verified. Based on (\ref{eq:Vdot2}), it can be seen that $\Gamma$ is the set such that:
 \begin{equation}\label{eq:inv}
 \begin{split}
    \dot{\alpha}=0,\dot{\beta}=0,\Rightarrow
    \ddot{\alpha}=0,\ddot{\beta}=0, \quad\quad  \\
    \dot{\gamma}=0,\dot{d}=0, \Rightarrow \ddot{\gamma}=0,\ddot{d}=0, \quad\quad \\ 
    \dot{\theta_1} =0,\dot{\theta_2} = 0, \Rightarrow \ddot{\theta_1} = 0, \ddot{\theta_2} = 0, \quad\\
    \dot{e}_{\alpha} = 0, \dot{e}_{\beta} = 0, \Rightarrow e_{\alpha} = \phi_1,e_{\beta} = \phi_2, \\
    \dot{e}_{\gamma} = 0, \dot{e}_{d} = 0, \Rightarrow e_{\gamma} = \phi_3 ,e_{d} = \phi_4, \\ 
\end{split}
\end{equation}

where $\phi_{1,2,3,4}$ are constants to be determined.

Combining (\ref{eq:inv}) with (\ref{eq:u1})-(\ref{eq:u4}) and (\ref{eq:modelmatrix}), we obtain the conditions:
\begin{equation}\label{syst_sol}
\begin{split}
k_{p\alpha}e_{\alpha} = 0,\quad\quad\quad\quad\quad \\
k_{p\beta}e_{\beta} = 0,\quad\quad\quad\quad\quad \\
k_{p\gamma}e_{\gamma} = 0,\quad\quad\quad\quad\quad \\
-gmcos\theta_1cos\theta_2 = -mg + k_{pd}e_{d}, \\
gmdcos\theta_2sin\theta_1 = 0,\quad\quad\quad\quad\quad \\
gmdcos\theta_1sin\theta_2 = 0.\quad\quad\quad\quad\quad\\
\end{split}
\end{equation}

From the first three equations of (\ref{syst_sol}),we will have that
\begin{equation}\label{eq:sola}
    k_{p\alpha}e_\alpha =  0  \Rightarrow e_\alpha = 0 \Rightarrow \phi_1 = 0 \Rightarrow \alpha = \alpha_d,  
\end{equation}
\begin{equation}\label{eq:solb}
    k_{p\beta}e_\beta =  0  \Rightarrow e_\beta = 0 \Rightarrow \phi_2 = 0 \Rightarrow \beta = \beta_d,  
\end{equation}
\begin{equation}\label{eq:solg}
    k_{p\gamma}e_\gamma =  0  \Rightarrow e_\gamma = 0 \Rightarrow \phi_3 = 0 \Rightarrow \gamma = \gamma_d.  
\end{equation}
From the last two equations of (\ref{syst_sol}), due to Assumption \ref{ass:2}, one can be obtained that:
\begin{equation}
\begin{split}
cos\theta_2sin\theta_1 = 0, \\
cos\theta_1sin\theta_2 = 0,\\
\end{split}
\end{equation}
The following conclusion can be achieved:
\begin{equation}\label{eq:solutionth1th2}
    \theta_1 = \theta_2 =  (k\pi) \lor {(2k+1)\over{2}}\pi, \quad k \in \mathbb{Z}.   
\end{equation}
However, due to Assumption \ref{ass:1}, the only acceptable solution will be:
\begin{equation}\label{eq:solutionth1th22}
    \theta_1 = \theta_2 =  0.  
\end{equation}
By inserting the (\ref{eq:solutionth1th22}) in the forth equation of (\ref{syst_sol}), one can conclude that:
\begin{equation}\label{eq:sold}
    k_{pd}e_d =  0  \Rightarrow e_d = 0 \Rightarrow \phi_4 = 0 \Rightarrow d = d_d,  
\end{equation}
\end{proof}

\begin{remark}
Other types of cranes such as overhead cranes, boom cranes, tower cranes etc, have similar dynamic characteristics to a knuckle cranes. Accordingly, the controller proposed in this
paper may be adapted to all these underactuated system.
\end{remark}

\section{Simulation Results}
In order to demonstrate the effectiveness of the proposed strategy, in this section we simulate the knuckle crane in Fig.~\ref{fig:crane}. The practical performances of the proposed control approach are compared with a linear quadratic regulator (LQR) obtained by linearization.The physical parameters are selected as follows: $m_b=2kg, \quad m_j=3kg, \quad m=1k, \quad l_b=5m, \quad l_j=4m$.
\newline
The parameters for the control law (\ref{eq:u1})-(\ref{eq:u4}) are the following:$ k_{p\alpha} = 30,k_{p\beta} = 10,k_{p\gamma} = 10,k_{pd} = 1,k_{d\alpha} = 50,k_{d\beta} = 30,k_{d\gamma} = 50,k_{dd} = 10.$
\newline
For the LQR control approach, first, the crane dynamics is linearized around the equilibrium point, and then, the following weight matrices have been chosen to stabilize the plant: $Q = diag\{25, 400, 450, 200, 1, 1, 1, 1, 1, 1, 120, 120\}$ and $R = diag\{0.1, 0.1, 0.1, 1\}$.

\medskip

As one can see in Figg.~\ref{fig:alpha}-\ref{fig:beta}-\ref{fig:gamma} both controllers successfully move the crane towards
the desired angular positions. 
Although the LQR controller moves the first three degrees of the cranes (tower, boom, and jib angles) slightly faster than the method presented in this paper, one can notice that our method produces much less oscillations of the load (see Figg.\ref{fig:th1}-\ref{fig:th2})  and of the cable (see Fig.~\ref{fig:d}) , resulting in an overall faster and safer stabilization of the load. It is worth noting that in both methods, the input profiles and maximum values are almost similar. This values are well within the typical limits of the crane actuators.

\begin{figure}[ht!]
\centering
\includegraphics[width=6cm,height=3cm]{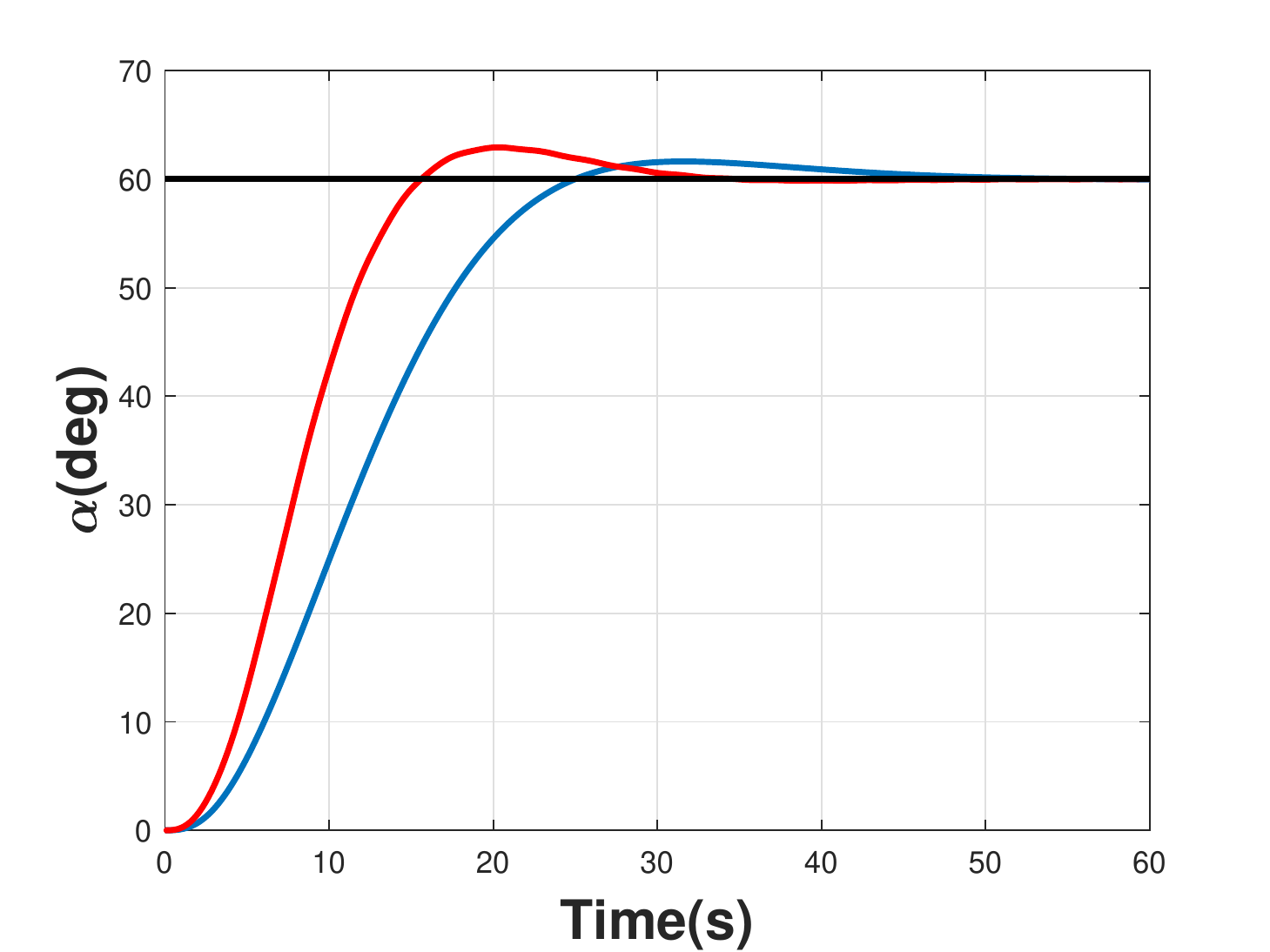}
\caption{\label{fig:alpha} Tower angle $\alpha$. Black line: Desired reference. Blue line: Nonlinear controller. Red line: LQR.}
\end{figure}

\begin{figure}[ht!]
\centering
\includegraphics[width=6cm,height=3cm]{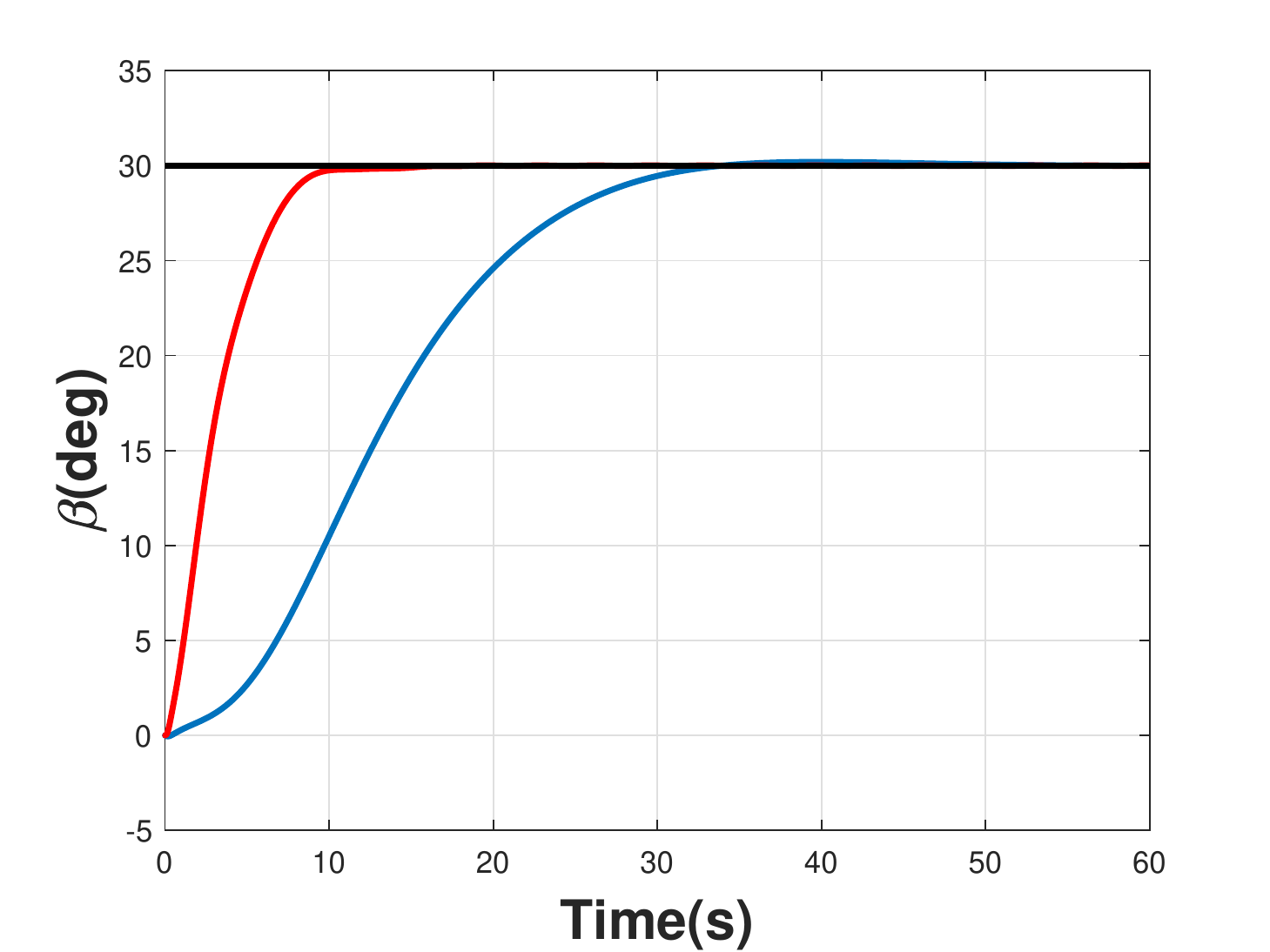}
\caption{\label{fig:beta} Boom angle $\beta$. Black line: Desired reference. Blue line: Nonlinear controller. Red line: LQR.}
\end{figure}

\begin{figure}[ht!]
\centering
\includegraphics[width=6cm,height=3cm]{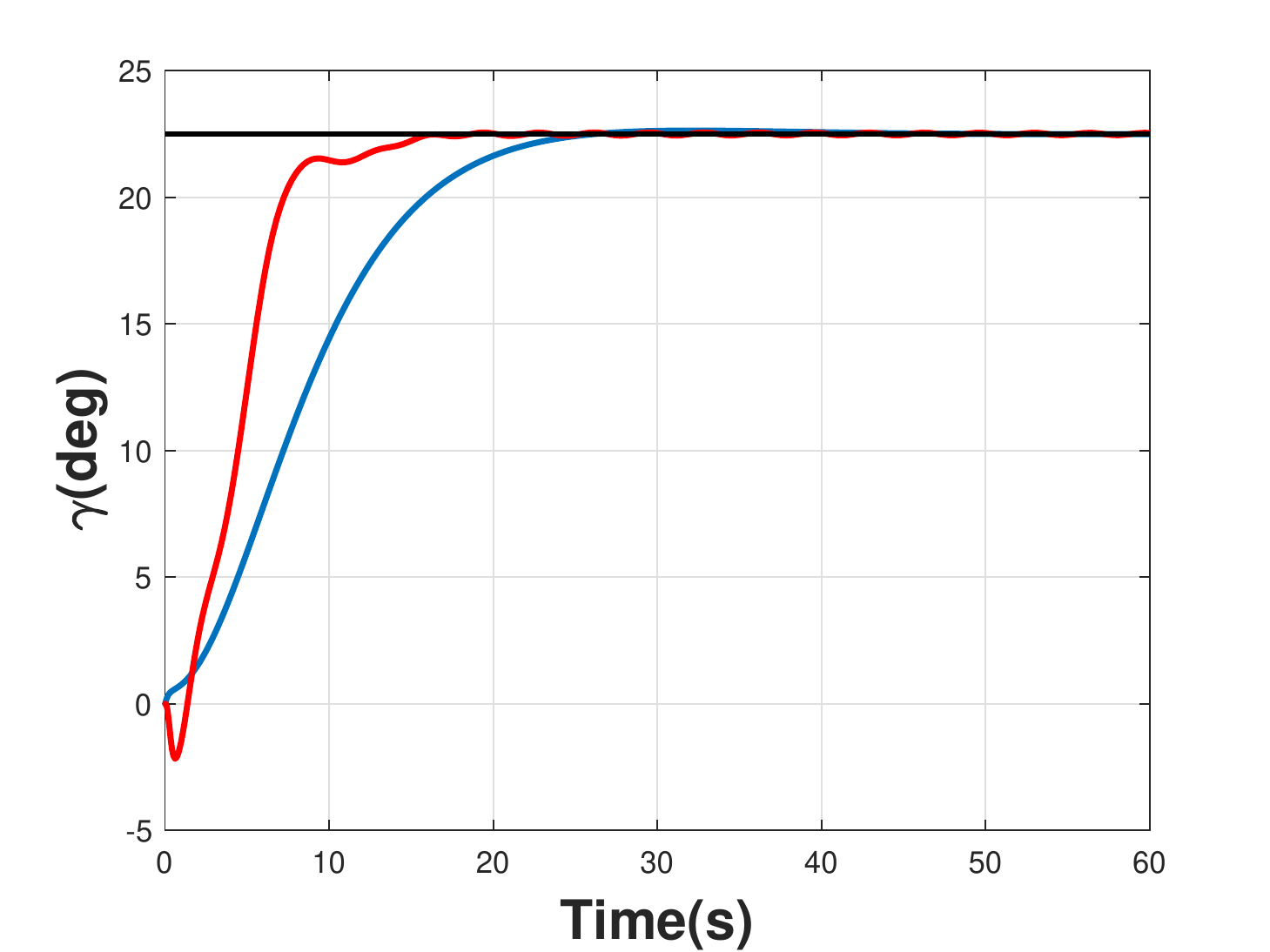}
\caption{\label{fig:gamma} Jib angle $\gamma$. Black line: Desired reference. Blue line: Nonlinear controller. Red line: LQR.}
\end{figure}

\begin{figure}[ht!]
\centering
\includegraphics[width=6cm,height=3cm]{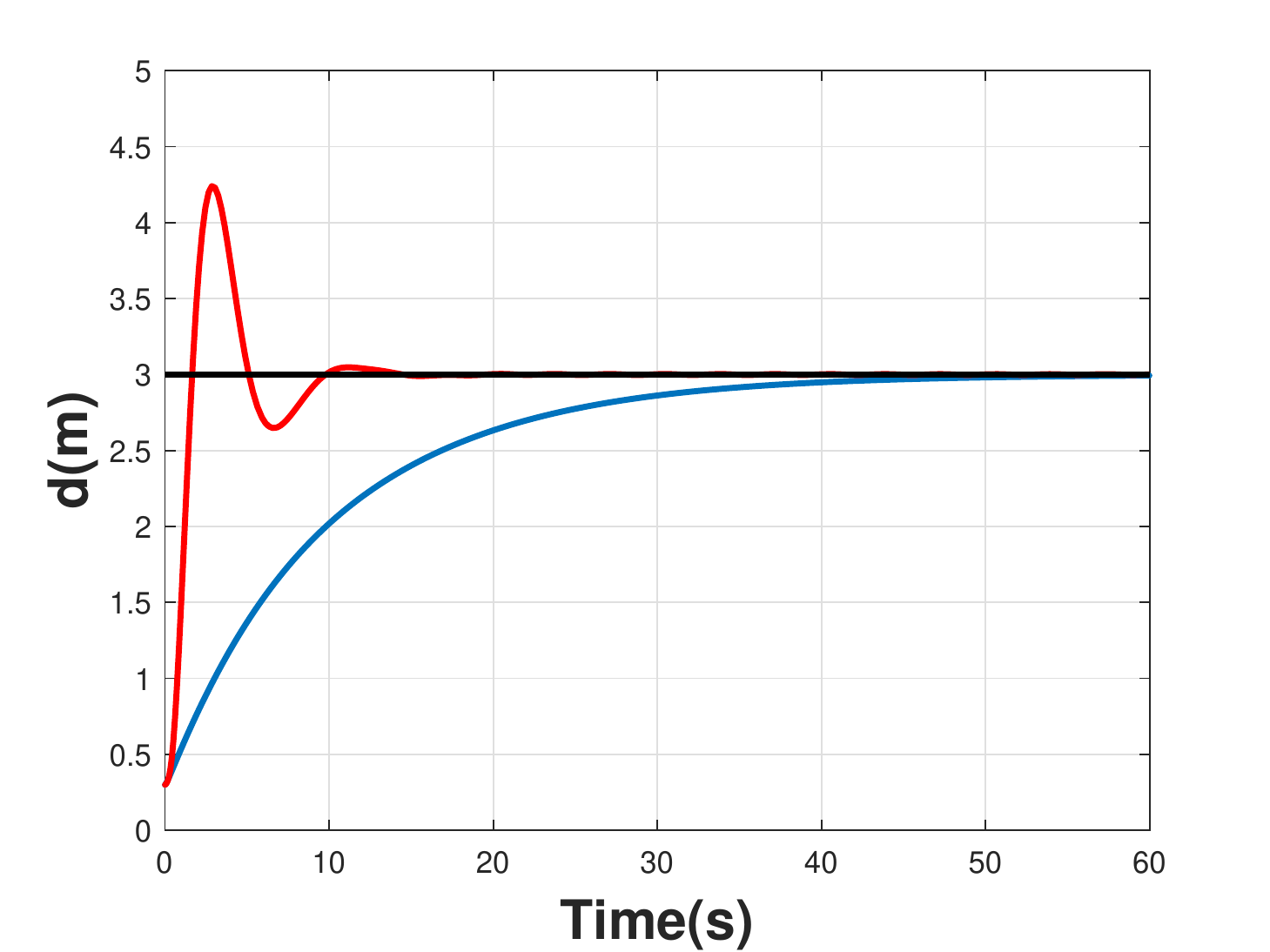}
\caption{\label{fig:d} Rope length. Black line: Desired reference. Blue line: Nonlinear controller. Red line: LQR.}
\end{figure}

\begin{figure}[ht!]
\centering
\includegraphics[width=6cm,height=3cm]{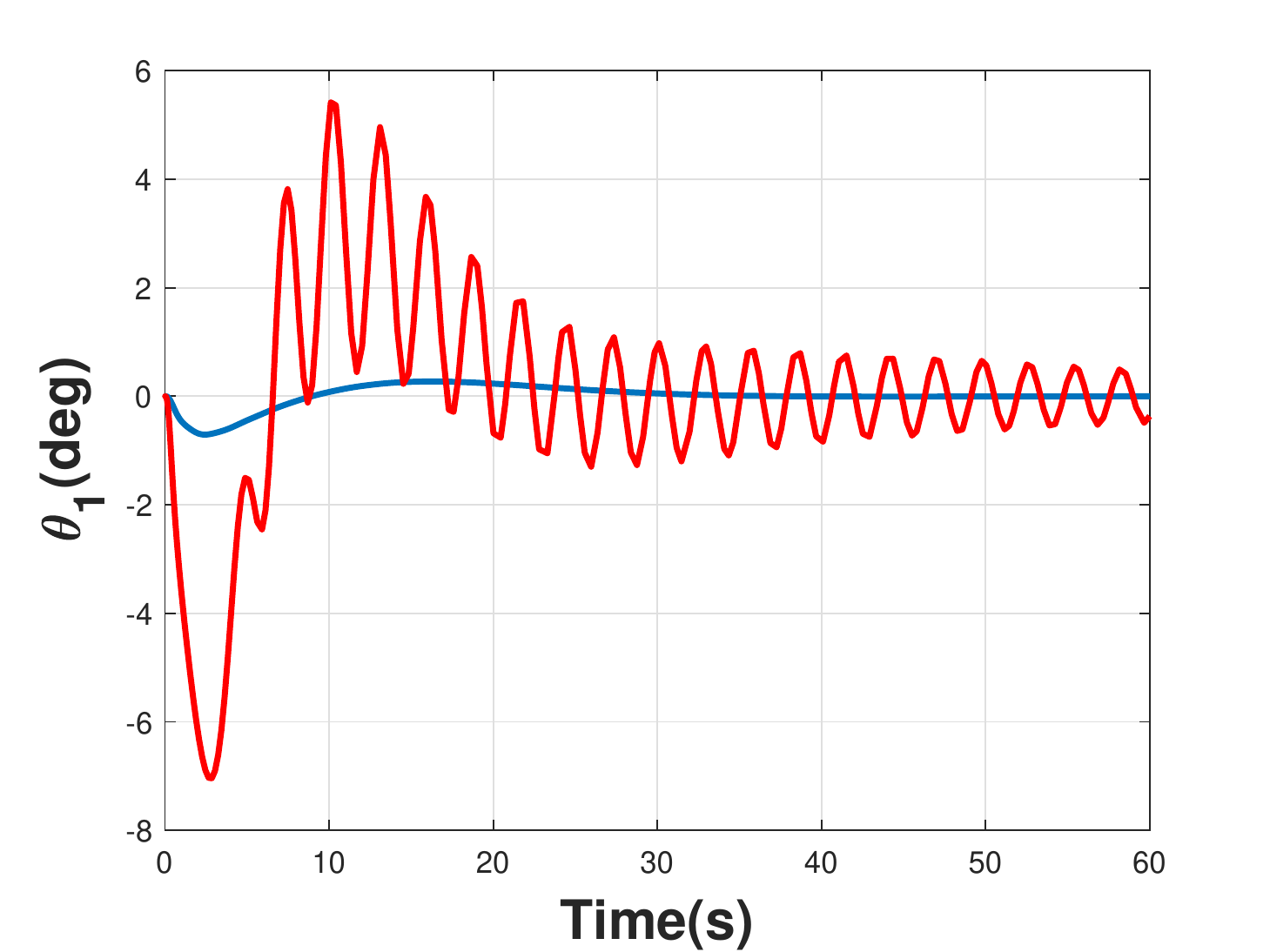}
\caption{\label{fig:th1} Payload angle  $\theta_1$. Blue line: Nonlinear controller. Red line: LQR.}
\end{figure}

\begin{figure}[ht!]
\centering
\includegraphics[width=6cm,height=3cm]{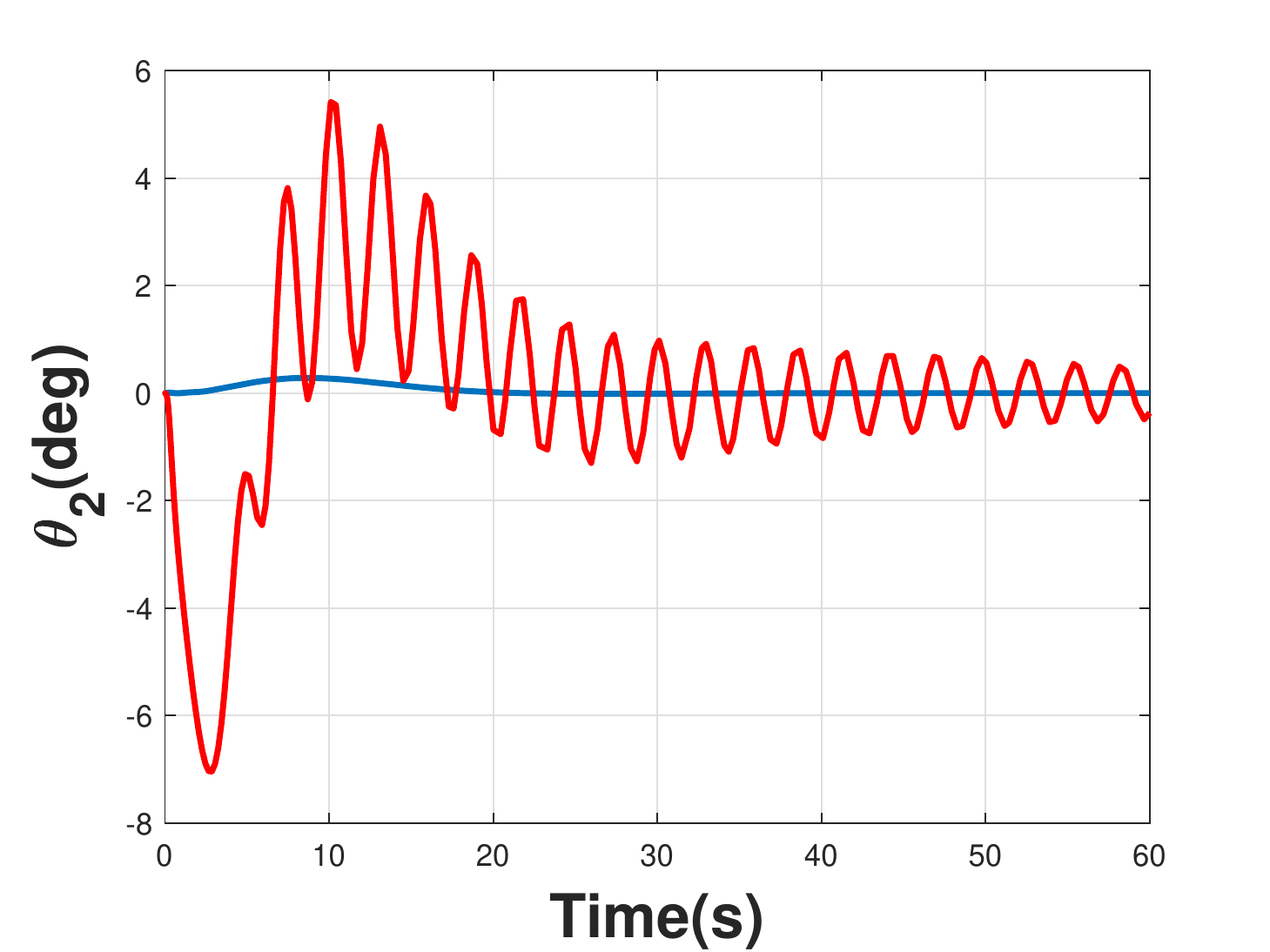}
\caption{\label{fig:th2} Payload angle  $\theta_2$ Blue line: Nonlinear controller. Red line: LQR.}
\end{figure}

\begin{figure}[ht!]
\centering
\includegraphics[width=8cm,height=4cm]{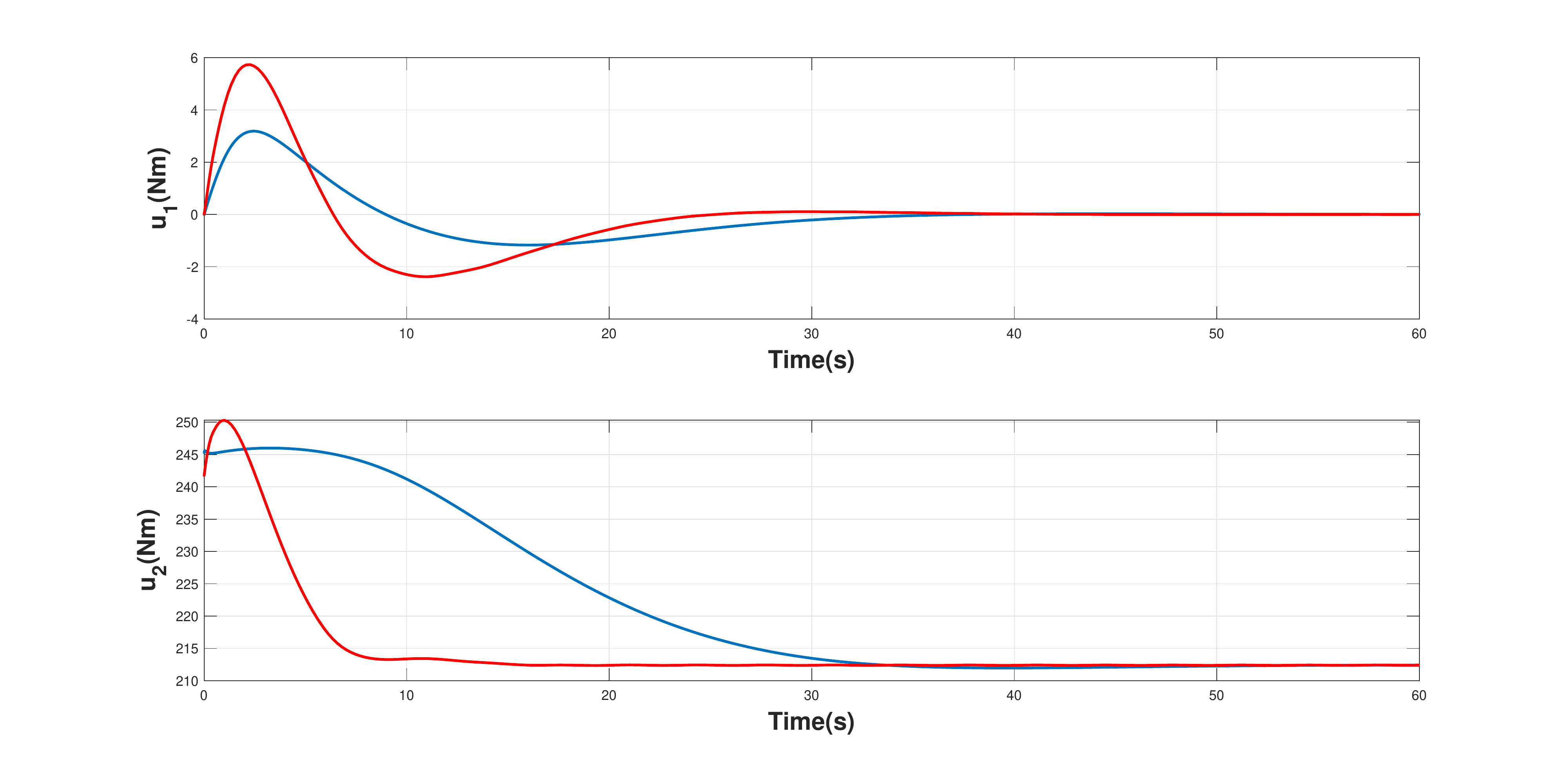}
\caption{\label{fig:u1} Control inputs $u_1~\&~u_2$. Blue line: Nonlinear controller. Red line: LQR.}
\end{figure}

\begin{figure}[ht!]
\centering
\includegraphics[width=8cm,height=4cm]{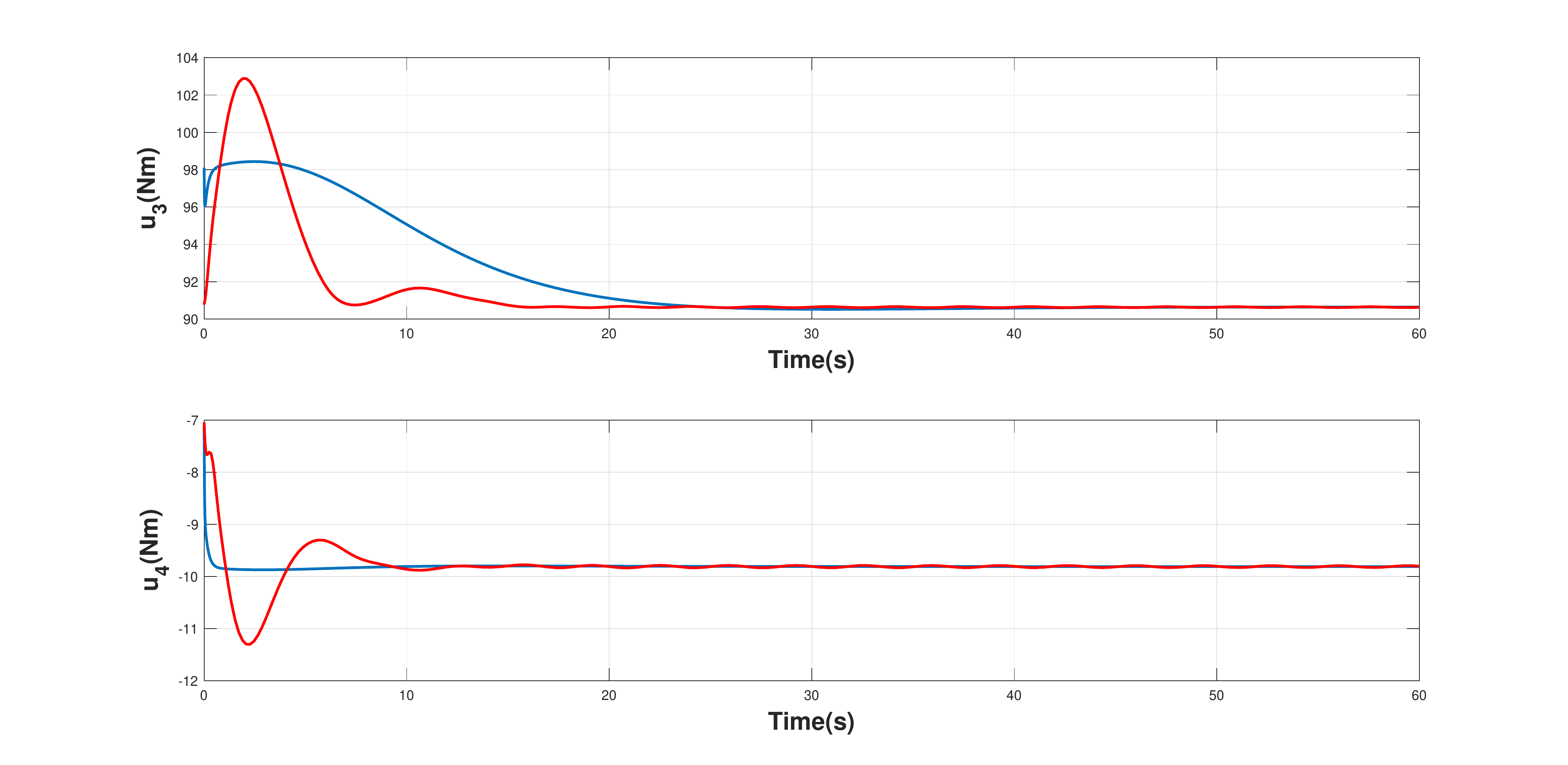}
\caption{\label{fig:u2} Control inputs $u_3~\&~u_4$. Blue line: Nonlinear controller. Red line: LQR.}
\end{figure}

\section{Conclusion}\label{sec:Conclusion}
This paper proposed a nonlinear control scheme for the control of knuckle crane. The main contribution of this article is that for the first time a detailed mathematical model is shown where the complexity of this type of system emerges. Furthermore, the proposed control is based directly on the nonlinear system avoiding linearization. Despite the complexity of the model, the proposed control scheme is able to guide the crane towards a desired reference and ensuring that the non-actuated variables (i.e., $\theta_1$ and $\theta_2$) go to zero in a fast way.

\section*{Appendix}
The nonzero entries of the system matrices $M(q)$, $C(q,\dot q)$, $F(\dot q)$ and $g(q)$ are define as follows:
\begin{align*}
m_{11} = I_{tot} + d^2m + A_1C_{\beta}^2 + A_2C_{\gamma}^2 + A_3C_{\beta}C_{\gamma}\\ 
+ 2A_4dC_{\beta}S_{\theta_2} + 2A_5dC_{\gamma}S_{\theta_2} - d^2mC_{\theta_1}^2C_{\theta_2}^2 \\
m_{22} = A_1 + I_B \quad m_{33} = A_2 + I_J \quad m_{44} = m \\
m_{55} = dmC_{\theta_2}^2 \quad m_{66} = dm \quad m_{12} = m_{21} = A_4dC_{\theta_2}S_{\beta}S_{\theta_1} \\
m_{13} = m_{31} = A_5dC_{\theta_2}S_{\gamma}S_{\theta_1} \quad
m_{14} = m_{41} = C_{\theta_2}S_{\theta_1}(A_4C_{\beta} + A_5C_{\gamma})\\
m_{15} = m_{51} = dC_{\theta_1}C_{\theta_2}(A_4C_{\beta} + A_5C_{\gamma} + dmS_{\theta_2})\\
m_{16} = m_{61} =-dS_{\theta_1}(dm + A_4C_{\beta}S_{\theta_2} + A_5C_{\gamma}S_{\theta_2})\\
m_{23} = m_{32} = {1\over{2}}(A_3C_{\beta-\gamma}),
m_{24} = m_{42} = -A_4(S_{\beta}S_{\theta_2} + C_{\beta}C_{\theta_1}C_{\theta_2})\quad \\
m_{25} = m_{52} = A_4dC_{\beta}C_{\theta_2}S_{\theta_1}\\
m_{26} = m_{62} = -A_4d(C_{\theta_2}S_{\beta} - C_{\beta}C_{\theta_1}S_{\theta_2}) \\
m_{34} = m_{43} = -A_5(S_{\gamma}S_{\theta_2} + C_{\gamma}C_{\theta_1}C_{\theta_2})\\
m_{35} = m_{53} = A_5dC_{\gamma}C_{\theta_2}S_{\theta_1}\quad,
m_{36} = m_{63} = -A_5d(C_{\theta_2}S_{\gamma} - C_{\gamma}C_{\theta_1}S_{\theta_2})\\
c_{11} = 2d\dot{d}m - A_2\dot{\gamma}S_{2\gamma} - A_1\dot{\beta}S_{2\beta} + 2A_4\dot{d}C_{\beta}S_{\theta_2} \\
+ 2A_5\dot{d}C_{\gamma}S_{\theta_2} - A_3\dot{\beta}C_{\gamma}S_{\beta}  - A_3\dot{\gamma}C_{\beta}S_{\gamma} \\
- 2d\dot{d}mC_{\theta_1}^2C_{\theta_2}^2 + 2A_4d\dot{\theta_2}C_{\beta}C_{\theta_2} + 2A_5d\dot{\theta_2}C_{\gamma}C_{\theta_2} 
- 2A_4d\dot{\beta}S_{\beta}S_{\theta_2}\\ - 2A_5d\dot{\gamma}S_{\gamma}S_{\theta_2} + 2d^2\dot{\theta_1}mC_{\theta_1}C_{\theta_2}^2S_{\theta_1}
+ 2d^2\dot{\theta_2}mC_{\theta_1}^2C_{\theta_2}S_{\theta_2}\\
c_{12} = A_4\dot{d}C_{\theta_2}S_{\beta}S_{\theta_1} + 
A_d\dot{\beta}C_{\beta}C_{\theta_2}S_{\theta_1} + \\ A_4d\dot{\theta_1}C_{\theta_1}C_{\theta_2}S_{\beta} - A_4d\dot{\theta_2}S_{\beta}S_{\theta_1}S_{\theta_2} \\
c_{13} = A_5\dot{d}C_{\theta_2}S_{\gamma}S_{\theta_1} + A_5d\dot{\gamma}C_{\gamma}C_{\theta_2}S_{\theta_1} + \\ A_5d\dot{\theta_1}C_{\theta_1}C_{\theta_2}S_{\gamma} - A_5d\dot{\theta_2}S_{\gamma}S_{\theta_1}S_{\theta_2} \\
c_{14} = \dot{\theta_1}C_{\theta_1}C_{\theta_2}(A_4C_{\beta} + A_5C_{\gamma}) - C_{\theta_2}S_{\theta_1}(A_4\dot{\beta}S_{\beta} + A_5\dot{\gamma}S_{\gamma}) \\ - \dot{\theta_2}S_{\theta_1}S_{\theta_2}(A_4C_{\beta} + A_5C_{\gamma}) \\
c_{15} = dC_{\theta_1}C_{\theta_2}(\dot{d}mS_{\theta_2} - A_4\dot{\beta}S_{\beta} - A_5d\dot{\gamma}S_{\gamma} +  d\dot{\theta_2}mC_{\theta_2}) \\ + \dot{d}C_{\theta_1}C_{\theta_2}(A_4C_{\beta} + A_5C_{\gamma} + dmS_{\theta_2}) - d\dot{\theta_1}C_{\theta_2}S_{\theta_1}(A_4C_{\beta} \\ + A_5C_{\gamma}  + dmS_{\theta_2}) - d\dot{\theta_2}C_{\theta_1}S_{\theta_2}(A_4C_{\beta} + A_5C_{\gamma} + dmS_{\theta_2}) \\
c_{16} = - dS_{\theta_1}(\dot{d}m + A_4\dot{\theta_2}C_{\beta}C_{\theta_2} + A_5\dot{\theta_2}C_{\gamma}C_{\theta_2} - \\ A_4d\dot{\beta}S_{\beta}S_{\theta_2} - A_5\dot{\gamma}S_{\gamma}S_{\theta_2}) - \dot{d}S_{\theta_1}(dm + A_4C_{\beta}S_{\theta_2} + \\ A_5C_{\gamma}S_{\theta_2}) - d\dot{\theta_1}C_{\theta_1}(dm + A_4C_{\beta}S_{\theta_2} + A_5C_{\gamma}S_{\theta_2}) \\
c_{21} = {1\over{2}}(\dot{\alpha}S_{\beta}(2A_1C_{\beta} + A_3C_{\gamma} + 2A_4dS_{\theta_2})) + \\ {1\over{2}}(3A_4\dot{d}C_{\theta_2}S_{\beta}S_{\theta_1}) +  {1\over{2}}(A_4d\dot{\beta}C_{\beta}C_{\theta_2}S_{\theta_1}) + \\ {1\over{2}}(3A_4d\dot{\theta_1}C_{\theta_1}C_{\theta_2}S_{\beta}) - {1\over{2}}(3A_4d\dot{\theta_2}S_{\beta}S_{\theta_1}S_{\theta_2}) \\
\end{align*}
\begin{align*}
c_{22} = {1\over{4}}(A_3\dot{\gamma}S_{\beta-\gamma}) + {1\over{2}}(A_4\dot{d}(C_{\beta}S_{\theta_2} -  C_{\theta_1}C_{\theta_2}S_{\beta})) \\ +  {1\over{2}}(A_4d\dot{\theta_2}(C_{\beta}C_{\theta_2}  + C_{\theta_1}S_{\beta}S_{\theta_2})) - {1\over{2}}(A_4d\dot{\alpha}C_{\beta}C_{\theta_2}S_{\beta}) \\ +
{1\over{2}}(A_4d\dot{\theta_1}C_{\theta_2}S_{\beta}S_{\beta}),\quad
c_{23} = -{1\over{4}}(A_3S_{\beta-\gamma})(\dot{\beta} - 2\dot{\gamma})) \\
c_{24} = {1\over{2}}(A_4\dot{\beta}C_{\theta_1}C_{\theta_2}S_{\beta}) - A_4\dot{\theta_2}C_{\theta_2}S_{\beta} - {1\over{2}}(A_4\dot{\beta}C_{\beta}S_{\theta_2}) \\ + A_4\dot{\theta_1}C_{\beta}C_{\theta_2}S_{\beta} + A_4\dot{\theta_2}C_{\beta}C_{\theta_1}S_{\theta_2} +  {1\over{2}}(A_4\dot{\alpha}C_{\theta_2}S_{\beta}S_{\beta}) \\
c_{25} =A_4\dot{d}C_{\beta}C_{\theta_2}S_{\beta} + {1\over{2}}(A_4d\dot{\alpha}C_{\theta_1}C_{\theta_2}S_{\beta}) \\ + 
A_4d\dot{\theta_1}C_{\beta}C_{\theta_1}C_{\theta_2}  -
{1\over{2}}(A_4d\dot{\beta}C_{\theta_2}S_{\beta}S_{\beta}) - A_4d\dot{\theta_2}C_{\beta}S_{\beta}S_{\theta_2} \\
c_{26} = A_4d\dot{\theta_2}S_{\beta}S_{\theta_2} - {1\over{2}}(A_4d\dot{\beta}C_{\beta}C_{\theta_2}) \\ - A_4\dot{d}C_{\theta_2}S_{\beta} +  A_4\dot{d}C_{\beta}C_{\theta_1}S_{\theta_2} \\ -  {1\over{2}}(A_4d\dot{\beta}C_{\theta_1}S_{\beta}S_{\theta_2}) - A_4d\dot{\theta_1}C_{\beta}S_{\beta}S_{\theta_2} \\ - {1\over{2}}(A_4d\dot{\alpha}S_{\beta}S_{\beta}S_{\theta_2}) + A_4d\dot{\theta_2}C_{\beta}C_{\theta_1}C_{\theta_2} \\
c_{31} = {1\over{2}}(\dot{\alpha}S_{\gamma}(2A_2C_{\gamma} + A_3C_{\beta} + 2A_5dS_{\theta_2}))  \\ + {1\over{2}}(3A_5\dot{d}C_{\theta_2}S_{\gamma}S_{\theta_1}) +  {1\over{2}}(A_5d\dot{\gamma}C_{\gamma}C_{\theta_2}S_{\theta_1}) \\ + {1\over{2}}(3A_5d\dot{\theta_1}C_{\theta_1}C_{\theta_2}S_{\gamma})  -{1\over{2}}(3A_5d\dot{\theta_2}S_{\gamma}S_{\theta_1}S_{\theta_2})\\
c_{32} = -{1\over{4}}(A_3S_{\beta-\gamma})(2\dot{\beta} - \dot{\gamma}))\\
c_{33} = {1\over{2}}(A_5\dot{d}(C_{\gamma}S_{\theta_2} - C_{\theta_1}C_{\theta_2}S_{\gamma}))  -{1\over{4}} (A_3\dot{\beta}S_{\beta-\gamma}) \\ + {1\over{2}}(A_5d\dot{\theta_2}(C_{\gamma}C_{\theta_2} + C_{\theta_1}S_{\gamma}S_{\theta_2})) - {1\over{2}}(A_5d\dot{\alpha}C_{\gamma}C_{\theta_2}S_{\theta_1}) \\ + {1\over{2}}(A_5d\dot{\theta_1}C_{\theta_2}S_{\gamma}S_{\theta_1})
\\
c_{34} = {1\over{2}}(A_5\dot{\gamma}C_{\theta_1}C_{\theta_2}S_{\gamma}) - A_5\dot{\theta_2}C_{\theta_2}S_{\gamma} - {1\over{2}}(A_5\dot{\gamma}C_{\gamma}S_{\theta_2}) \\+ A_5\dot{\theta_1}C_{\gamma}C_{\theta_2}S_{\theta_1} + A_5\dot{\theta_2}C_{\gamma}C_{\theta_1}S_{\theta_2} + {1\over{2}}(A_5\dot{\alpha}C_{\theta_2}S_{\gamma}S_{\theta_1})\\
c_{35} = A_5ddC_{\gamma}C_{\theta_2}S_{\theta_1} + {1\over{2}}(A_5d\dot{\alpha}C_{\theta_1}C_{\theta_2}S_{\gamma}) \\ - {1\over{2}}(A_5d\dot{\gamma}C_{\theta_2}S_{\gamma}S_{\theta_1}) - A_5d\dot{\theta_2}C_{\gamma}S_{\theta_1}S_{\theta_2} + A_5d\dot{\theta_1}C_{\gamma}C_{\theta_1}C_{\theta_2}
\\
c_{36} = A_5d\dot{\theta_2}S_{\gamma}S_{\theta_2} - {1\over{2}}(A_5d\dot{\gamma}C_{\gamma}C_{\theta_2}) - A_5\dot{d}C_{\theta_2}S_{\gamma} \\ + A_5\dot{d}C_{\gamma}C_{\theta_1}S_{\theta_2} -{1\over{2}} (A_5d\dot{\gamma}C_{\theta_1}S_{\gamma}S_{\theta_2}) \\- A_5d\dot{\theta_1}C_{\gamma}S_{\theta_1}S_{\theta_2} -{1\over{2}} (A_5d\dot{\alpha}S_{\gamma}S_{\theta_1}S_{\theta_2}) + A_5d\dot{\theta_2}C_{\gamma}C_{\theta_1}C_{\theta_2}\\
c_{41} = d\dot{\theta_2}mS_{\theta_1} - d\dot{\alpha}m - A_4\dot{\alpha}C_{\beta}S_{\theta_2} - A_5\dot{\alpha}C_{\gamma}S_{\theta_2} \\+ d\dot{\alpha}mC_{\theta_1}^2C_{\theta_2}^2 + {1\over{2}}(A_4\dot{\theta_1}C_{\beta}C_{\theta_1}C_{\theta_2}) + {1\over{2}}(A_5\dot{\theta_1}C_{\gamma}C_{\theta_1}C_{\theta_2}) \\-{1\over{2}}(3A_4\dot{\beta}C_{\theta_2}S_{\beta}S_{\theta_1}) - {1\over{2}}(A_4\dot{\theta_2}C_{\beta}S_{\theta_1}S_{\theta_2}) \\ - {1\over{2}}(3A_5\dot{\gamma}C_{\theta_2}S_{\gamma}S_{\theta_1})-{1\over{2}}(A_5\dot{\theta_2}C_{\gamma}S_{\theta_1}S_{\theta_2})-d\dot{\theta_1}mC_{\theta_1}C_{\theta_2}S_{\theta_2}\\
c_{42} = A_4\dot{\beta}C_{\theta_1}C_{\theta_2}S_{\beta} - (A_4\dot{\theta_2}C_{\theta_2}S_{\beta})/2 - A_4\dot{\beta}C_{\beta}S_{\theta_2} \\+ {1\over{2}}(A_4\dot{\theta_1}C_{\beta}C_{\theta_2}S_{\theta_1}) + {1\over{2}}(A_4\dot{\theta_2}C_{\beta}C_{\theta_1}S_{\theta_2}) - {1\over{2}}(A_4\dot{\alpha}C_{\theta_2}S_{\beta}S_{\theta_1}) \\
\end{align*}
\begin{align*}
c_{43} = A_5\dot{\gamma}C_{\theta_1}C_{\theta_2}S_{\gamma} - {1\over{2}}(A_5\dot{\theta_2}C_{\theta_2}S_{\gamma})  \\ - A_5\dot{\gamma}C_{\gamma}S_{\theta_2} + {1\over{2}}(A_5\dot{\theta_1}C_{\gamma}C_{\theta_2}S_{\theta_1})  \\ + {1\over{2}}(A_5\dot{\theta_2}C_{\gamma}C_{\theta_1}S_{\theta_2}) -
{1\over{2}}(A_5\dot{\alpha}C_{\theta_2}S_{\gamma}S_{\theta_1}) \\
c_{45} = d\dot{\theta_1}m(S_{\theta_2}^2 - 1) - {1\over{2}}(\dot{\alpha}C_{\theta_1}C_{\theta_2}(A_4C_{\beta} \\ + A_5C_{\gamma} + 2dmS_{\theta_2})) - {1\over{2}}(A_4\dot{\beta}C_{\beta}C_{\theta_2}S_{\theta_1})  \\ - {1\over{2}}(A_5\dot{\gamma}C_{\gamma}C_{\theta_2}S_{\theta_1})\\
c_{46} = {1\over{2}}(A_4\dot{\beta}(C_{\theta_2}S_{\beta} - C_{\beta}C_{\theta_1}S_{\theta_2})) \\ + {1\over{2}}(A_5\dot{\gamma}(C_{\theta_2}S_{\gamma} - C_{\gamma}C_{\theta_1}S_{\theta_2})) \\ - d\dot{\theta_2}m + {1\over{2}}(\dot{\alpha}S_{\theta_1}(2dm \\ + A_4C_{\beta}S_{\theta_2} + A_5C_{\gamma}S_{\theta_2}))  \\
c_{51} = 2d^2\dot{\theta_2}mC_{\theta_1}C_{\theta_2}^2 - {1\over{2}}(d^2\dot{\theta_2}mC_{\theta_1}) \\ +  {1\over{2}}(A_4\dot{d}C_{\beta}C_{\theta_1}C_{\theta_2}) + {1\over{2}}(A_5\dot{d}C_{\gamma}C_{\theta_1}C_{\theta_2}) \\ - {1\over{2}}(3A_4d\dot{\beta}C_{\theta_1}C_{\theta_2}S_{\beta})  - {1\over{2}}(A_4d\dot{\theta_1}C_{\beta}C_{\theta_2}S_{\theta_1}) \\ - {1\over{2}}(A_4d\dot{\theta_2}C_{\beta}C_{\theta_1}S_{\theta_2})  - {1\over{2}}(3A_5d\dot{\gamma}C_{\theta_1}C_{\theta_2}S_{\gamma}) \\ - {1\over{2}}(A_5d\dot{\theta_1}C_{\gamma}C_{\theta_2}S_{\theta_1}) \\ - {1\over{2}}(A_5d\dot{\theta_2}C_{\gamma}C_{\theta_1}S_{\theta_2}) + 2d\dot{d}mC_{\theta_1}C_{\theta_2}S_{\theta_2} \\ - {1\over{2}}(d^2\dot{\theta_1}mC_{\theta_2}S_{\theta_1}S_{\theta_2}) - d^2\dot{\alpha}mC_{\theta_1}C_{\theta_2}^2S_{\theta_1} \\
c_{52} = {1\over{2}}(A_4\dot{d}C_{\beta}C_{\theta_2}S_{\theta_1}) - {1\over{2}}(A_4d\dot{\alpha}C_{\theta_1}C_{\theta_2}S_{\beta}) \\- A_4d\dot{\beta}C_{\theta_2}S_{\beta}S_{\theta_1} - {1\over{2}}(A_4d\dot{\theta_2}C_{\beta}S_{\theta_1}S_{\theta_2}) \\ + {1\over{2}}(A_4d\dot{\theta_1}C_{\beta}C_{\theta_1}C_{\theta_2}) \\
c_{53} = {1\over{2}}(A_5\dot{d}C_{\gamma}C_{\theta_2}S_{\theta_1}) - {1\over{2}}(A_5d\dot{\alpha}C_{\theta_1}C_{\theta_2}S_{\gamma}) \\ - A_5d\dot{\gamma}C_{\theta_2}S_{\gamma}S_{\theta_1}  - {1\over{2}}(A_5d\dot{\theta_2}C_{\gamma}S_{\theta_1}S_{\theta_2}) \\ + {1\over{2}}(A_5d\dot{\theta_1}C_{\gamma}C_{\theta_1}C_{\theta_2}) \\
c_{54} = - {1\over{2}}(\dot{\alpha}C_{\theta_1}C_{\theta_2}(A_4C_{\beta} + A_5C_{\gamma})) \\ - {1\over{2}}(A_4\dot{\beta}C_{\beta}C_{\theta_2}S_{\theta_1}) - {1\over{2}}(A_5\dot{\gamma}C_{\gamma}C_{\theta_2}S_{\theta_1})\\
c_{55} = {1\over{2}}(d\dot{\alpha}C_{\theta_2}S_{\theta_1}(A_4C_{\beta} + A_5C_{\gamma} + dmS_{\theta_2})) \\ - {1\over{2}}(A_4d\dot{\beta}C_{\beta}C_{\theta_1}C_{\theta_2}) - {1\over{2}}(A_5d\dot{\gamma}C_{\gamma}C_{\theta_1}C_{\theta_2}) \\
c_{56} = {1\over{2}}(d\dot{\alpha}C_{\theta_1}(dm + A_4C_{\beta}S_{\theta_2} + A_5C_{\gamma}S_{\theta_2})) \\+ {1\over{2}}(A_4d\dot{\beta}C_{\beta}S_{\theta_1}S_{\theta_2}) + {1\over{2}}(A_5d\dot{\gamma}C_{\gamma}S_{\theta_1}S_{\theta_2}) \end{align*}
\begin{align*}
c_{61} =  {1\over{2}}(3A_4d\dot{\beta}S_{\beta}S_{\theta_1}S_{\theta_2}) - {1\over{2}}(d^2\dot{\theta_1}mC_{\theta_1})\\ - A_4d\dot{\alpha}C_{\beta}C_{\theta_2} - A_5d\dot{\alpha}C_{\gamma}C_{\theta_2} - {1\over{2}}(A_4d\dot{d}C_{\beta}S_{\theta_1}S_{\theta_2}) \\ - {1\over{2}}(A_5d\dot{d}C_{\gamma}S_{\theta_1}S_{\theta_2})  - {1\over{2}}(A_4d\dot{\theta_1}C_{\beta}C_{\theta_1}S_{\theta_2}) \\ - {1\over{2}}(A_4d\dot{\theta_2}C_{\beta}C_{\theta_2}S_{\theta_1})  - {1\over{2}}(A_5d\dot{\theta_1}C_{\gamma}C_{\theta_1}S_{\theta_2})\\ - {1\over{2}}(A_5d\dot{\theta_2}C_{\gamma}C_{\theta_2}S_{\theta_1}) - 2d\dot{d}mS_{\theta_1}+ {1\over{2}}(3A_5d\dot{\gamma}S_{\gamma}S_{\theta_1}S_{\theta_2}) \\
c_{62} = {1\over{2}}(A_4d\dot{\theta_2}S_{\beta}S_{\theta_2}) - A_4d\dot{\beta}C_{\beta}C_{\theta_2} -\  {1\over{2}}(A_4d\dot{d}C_{\theta_2}S_{\beta}) \\+ {1\over{2}}(A_4d\dot{d}C_{\beta}C_{\theta_1}S_{\theta_2}) - {1\over{2}}(A_4d\dot{\theta_1}C_{\beta}S_{\theta_1}S_{\theta_2}) \\ + {1\over{2}}(A_4d\dot{\alpha}S_{\beta}S_{\theta_1}S_{\theta_2}) + {1\over{2}}(A_4d\dot{\theta_2}C_{\beta}C_{\theta_1}C_{\theta_2})\\
c_{63} = {1\over{2}}(A_5d\dot{\theta_2}S_{\gamma}S_{\theta_2}) - A_5d\dot{\gamma}C_{\gamma}C_{\theta_2} - {1\over{2}}(A_5d\dot{d}C_{\theta_2}S_{\gamma}) \\ +{1\over{2}}(A_5d\dot{d}C_{\gamma}C_{\theta_1}S_{\theta_2}) - A_5d\dot{\gamma}C_{\theta_1}S_{\gamma}S_{\theta_2} \\ + {1\over{2}}(A_5d\dot{\alpha}S_{\gamma}S_{\theta_1}S_{\theta_2}) + {1\over{2}}(A_5d\dot{\theta_2}C_{\gamma}C_{\theta_1}C_{\theta_2})\\
c_{64} = {1\over{2}}(A_4\dot{\beta}(C_{\theta_2}S_{\beta} - C_{\beta}C_{\theta_1}S_{\theta_2})) + {1\over{2}}(A_5\dot{\gamma}(C_{\theta_2}S_{\gamma} \\ - C_{\gamma}C_{\theta_1}S_{\theta_2})) + {1\over{2}}(\dot{\alpha}S_{\theta_1}S_{\theta_2}(A_4C_{\beta} + A_5C_{\gamma})) \\
c_{65} = {1\over{2}}(d^2\dot{\alpha}mC_{\theta_1}) - d^2\dot{\alpha}mC_{\theta_1}C_{\theta_2}^2  +  {1\over{2}}(A_4d\dot{\alpha}C_{\beta}C_{\theta_1}S_{\theta_2}) \\ + {1\over{2}}(A_4d\dot{\beta}C_{\beta}S_{\theta_1}S_{\theta_2}) + {1\over{2}}(A_5d\dot{\gamma}C_{\gamma}S_{\theta_1}S_{\theta_2}) \\
c_{66} = {1\over{2}}(d\dot{\alpha}C_{\theta_2}S_{\theta_1}(A_4C_{\beta} + A_5C_{\gamma})) - \\ {1\over{2}}(A_5d\dot{\gamma}(S_{\gamma}S_{\theta_2} + C_{\gamma}C_{\theta_1}C_{\theta_2})) \\ - {1\over{2}}(A_4d\dot{\beta}(S_{\beta}S_{\theta_2} + C_{\beta}C_{\theta_1}C_{\theta_2}))\\
g_2 = {1\over{2}}gl_Bcos\beta(2m+m_B+2m_J), \\ 
g_3 = {1\over{2}}gl_Jcos\gamma(2m+m_J), \quad
g_4 = -gmcos\theta_1cos\theta_2,\\
g_5 = gmdcos\theta_2sin\theta_1,\quad 
g_6 = gmdcos\theta_1sin\theta_2 \\
f_1 = d_{\theta_1}C_{\theta_2}^2\lvert \theta_1\rvert\ \dot{\theta}_1, \quad\quad
f_2 = d_{\theta_2} \lvert \theta_2\rvert\ \dot{\theta_2}
\end{align*}
where the following auxiliary variables are defined: $
    A_1 = l_B^2m + (l_B^2m_B)/4+l_B^2 m_J,\quad A_2 = l_J^2m + (l_J^2 m_J)/4,
    A_3 = 2l_Bl_J m+ l_Bl_J m_J,\quad A_4 = 2l_B m,\quad A_5 = 2l_J m $

\bibliography{ISARC}

\end{document}